\documentclass{article}
\usepackage{ijcai22}

\usepackage{times}
\usepackage{color,xcolor}
\usepackage{soul}
\usepackage{url}
\usepackage[hidelinks]{hyperref}
\usepackage[utf8]{inputenc}
\usepackage[small]{caption}
\usepackage{graphicx}
\usepackage{amsmath}
\usepackage{amsthm}
\usepackage{booktabs}
\usepackage{framed}
\usepackage{algorithm}
\usepackage{algorithmic}
\newtheorem{theorem}{Theorem}
\newtheorem{lemma}{Lemma}
\newtheorem{defn}{Definition}
\newtheorem{prop}{Proposition}

\title{Multi-Unit Diffusion Auctions with Intermediaries}

\author{
	Bin Li$^1$\and
	Dong Hao$^2$\and
	Dengji Zhao$^3$\\
	\affiliations
	{\small{$^1$School of Computer Science \& Engineering, Nanjing University of Science and Technology\\
			$^2$School of Computer Science \& Engineering, University of Electronic Science and Technology of China\\
			$^3$School of Information Science and Theconology, ShanghaiTech University\\}}
	\emails
	\{cs.libin@njust,
	haodong@uestc,
	zhaodj@shanghaitech\}.edu.cn
}

\begin{document}

\maketitle
\begin{abstract}

This paper studies multi-unit auctions powered by intermediaries, where each intermediary owns a private set of unit-demand buyers and all intermediaries are networked with each other. Our goal is to incentivize the intermediaries to diffuse the auction information to individuals they can reach, including their private buyers and neighboring intermediaries, so that more potential buyers are able to participate in the auction. To this end, we build a diffusion-based auction framework which incorporates the strategic interaction of intermediaries. It is showed that the classic Vickrey-Clarke-Groves (VCG) mechanism within the framework can achieve the maximum social welfare, but it may decrease the seller's revenue or even lead to a deficit. To overcome the revenue issue, we propose a novel auction, called critical neighborhood auction, which not only maximizes the social welfare, but also improves the seller's revenue comparing to the VCG mechanism with/without intermediaries.
\end{abstract}
\section{Introduction}\label{sec:intro}
Intermediaries are everywhere affecting our lives. When you rent an apartment or buy something, the apps on your phone, like Airbnb or eBay, serve as intermediate platforms, matching you with qualified landlords, your search with products probably purchased. Even when you drive home, Uber may push you a ride-sharing request. These intermediary agencies, which have become an integral part of modern business, take advantage of big data and information technologies to reduce the search costs between suppliers and consumers, promote commodity circulation, and increase trade efficiency. Instead of taking the ownership of the products sold, the intermediaries, like the apps mentioned above, just bring the buyers and sellers together to make a deal and collect commissions from successful transactions. In this work, we consider auction markets powered by intermediaries, where each intermediary owns a private set of buyers and all intermediaries are networked with each other. With the aid of intermediaries, the seller can recruit more buyers to participate in the auction, and then improve her revenue and the allocation efficiency. However, the intermediaries are often partially accessible to the seller. Without using promotions, the seller can only reach a portion of intermediaries existing in the market, which will result in a local resource allocation and a potential revenue loss.

One way to tackle this problem is to advertise the sale to attract more participants. However, the inadequate exposure of the advertisements makes the return unpredictable. The seller may lose the investments if the advertisements do not bring valuable participants. Besides, it is not a trivial matter to carry out an advertisement campaign for small and local business. Instead, building the promotion inside the marketplace is an attractive marketing mode. The goal is to incentivize the early participants to further share the auction information to other individuals such that the auction information can be fully spread in the market.
Traditional marketing strategies have provided solutions for promoting a product or an innovation in social networks. For example, in multi-level marketing or viral marketing \cite{leskovec2007dynamics} the customers can receive appropriate rebates on future purchases or direct monetary rewards \cite{emek2011mechanisms,kleinberg2005query} if their friends have made any purchase on their recommendation. Nevertheless, these methods are not applicable in the intermediary-based auction market for two reasons. Firstly, they do not consider the potential costs in the transactions, which may include the commissions the intermediaries request, the information handling fee or transportation/labor costs for delivering the commodities, etc. Secondly, these methods also ignore the competition between buyers which can be an important part of the deal, especially for scarce resources.

Recently, diffusion auction design has received a lot of attention in the literature of mechanism design. In the seminal work \cite{li2017mechanism}, the authors built an auction framework for a seller selling one commodity in social networks, which combines both the information diffusion and the buyer's competition. They also proposed a novel auction mechanism under the framework, called the information diffusion mechanism, to incentivize the buyers to invite other buyers to the auction. After their work, many efforts have devoted to this thread from different angles. Some extended the basic single-item diffusion auction model to multi-unit setting \cite{Zhao2018Multi,kawasaki2020strategy} and weighted graphs \cite{li2019graph,Li2018CustomerSI}; some studied diffusion auctions constrained with additional properties, like fairness \cite{zhang2019incentivize}, collusion-proofness \cite{jeong2020groupwise} or resale-proofness \cite{LI2022103631}. Inspired by the diffusion auction model, there are papers that introduce diffusion incentives to crowdsourcing \cite{zhang2020collaborative,zhang2020sybil}, procurement \cite{moustafa2021diffusion,liu2021budget}, exchange \cite{kawasaki2021mechanism}, and fog computing \cite{yang2018incentive}. See \cite{zhao2021social,guo2021emerg} for the recent progress.

We contribute to the field of diffusion auction design and study diffusion auctions with intermediaries, taking into consideration the diffusion incentives of intermediaries, the potential transaction costs and the competition between buyers. The most related work is from \cite{Li2018CustomerSI}, where the authors investigated a similar setting for the case of single commodity. We generalize their model to the intermediary-based market with multiple commodities. Our main contribution is a multi-unit diffusion auction, called critical neighborhood auction, which not only maximizes the social welfare, but also optimizes the seller's revenue.


\section{Preliminaries}\label{sec:model}
We first illustrate the components of an auction market with intermediaries, then build a formal auction framework within the market.
\subsection{The Intermediary Model}
Let $s$ denote a seller who is willing to sell a set of $K\ge 1$ identical commodities. Besides the seller, the market consists of a set of agents, denoted by $N$, which are divided into two disjoint categories: a set of intermediaries $I$ and a set of buyers $B$. Each intermediary $i\in I$ owns a private set of buyers, and all intermediaries in the market are networked with each other. Let $r_i\subseteq N\setminus\{i\}$ denote the agents with whom $i$ can communicate in the market. For convenience, we use $r_s$ to denote the agents directly reached by $s$. 
Each buyer $j\in B$ in the market is unit-demand, meaning that she has use for only one commodity and her value, denoted by $v_j$, for consuming one or more commodities is unchanged.
A transaction is defined by an agent path $\{s=a_0, a_1, \cdots , a_m, j=a_{m+1}\}$, where $j$ is a winning buyer, $a_1,\cdots,a_m$ are successive intermediaries with $a_k\in r_{a_{k-1}}$ for $k=1,\cdots, m+1$. Note that there is no need for an intermediary when the winning buyer is the seller's direct neighbor.
Given an intermediary $i\in I$ and a neighbor $k\in r_i$, 
we use $c_{ik}$ to denote the transaction cost between agents $i$ and $k$. 
For example, $c_{ik}$ could be the commissions $i$ requests for processing the transaction with $k$ or be the transportation expenses of delivering the commodities from $i$ to $k$, etc. In this work, we consider separable costs with the form of $c_{ik}=w_{ik}n_{ik}$, where $w_{ik}$ represents the cost per transaction and $n_{ik}$ is the number of transactions involving $(i,k)$. In addition, we make two assumptions on the market model: $1)$ for every $i\in I$ and $j\in B$, $r_i$ and $v_j$ are private information, and $2)$ $w_{ik}$ is fixed and known for every $k\in r_i$.
In the following contents, we will slightly abuse notations and refer to $i$ as an intermediary, $j$ as a buyer, and $k$ as an arbitrary agent.

Figure \ref{example} demonstrates an example of the auction market, where the squares (except for the seller labeled by $s$) represents for the intermediaries and the circles are potential buyers. The number in each circle denotes the buyer's private valuation and the cost per transaction is labeled on each edge, and there is an edge $(i,k)$ if $k\in r_i$. 

\begin{figure}[t]
	\centering
	\includegraphics[width=3.3in]{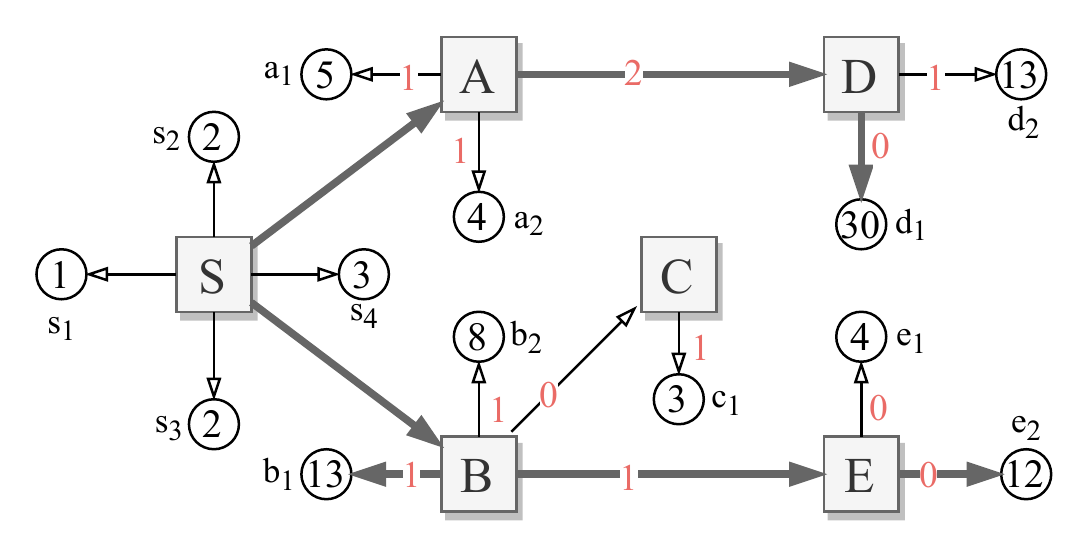}
	\caption{An example of the auction market with intermediaries. The subtree with bold lines represents for the efficient allocation graph (defined in Section \ref{sec:cnm}) with three commodities.}\label{example}
\end{figure}

\subsection{Auction Model with Intermediaries}

We next present the auction model with intermediaries. As usual, let $\theta_k$ be the private type of agent $k$ and $\Theta_k$ denote $k$'s type space. If $k$ is an intermediary, then $\theta_k=r_k$ and $\Theta_k=\mathcal{P}(N)$ where $\mathcal{P}(N)$ is the power set of $N$; otherwise, $\theta_k=v_k$ and $\Theta_k=\mathcal{R}_{+}$ for a buyer $k$. We use $\Theta=\times \Theta_{k\in N}$ to denote the type profile space of all agents. Since $\theta_k$ is private information, agent $k$ can cheat the mechanism to benefit herself via strategic report. Accordingly, let $\theta_k'$ be $k$'s reported type, where $\theta_k'=v_k'\in \mathcal{R}_{+}$ represents buyer $k$'s bid and $\theta_k'=r_k'\subseteq r_k$ is the neighbor set intermediary $k$ declares to have. Since each intermediary is only aware of her neighbors in the market, the misreport space of $r_k'$ is limited to $\mathcal{P}({r_k})$. Let ${\bf \theta'}$ be the reported type profile of all agents, ${\bf \theta'}_{-k}$ be the reported type profile of all agents except $k$, i.e., ${\bf \theta'} = (\theta_k', {\bf \theta'}_{-k})$.

Note that $r_k$ is private information, thus an agent can join in the sale only if they have received the auction information from someone who is already in the auction. In Figure \ref{example}, if agent $B$ does not share the auction information to agent $E$, then both $E$ and her private buyer set $\{e_1, e_2\}$ cannot participate in the sale.

\begin{defn}\label{valid_agent}
	Given a reported type profile ${\bf \theta}'$, we say agent $k$ is valid if there exists a sequence of intermediaries $\{a_1,a_2,\cdots,a_m\}$ with $a_1\in r_s, a_i\in r'_{a_{i-1}},k\in r'_{a_m}$ for $1<i\leq m$.
\end{defn}

That is, $k$ is a valid agent if there is a ``diffusion path" from the seller $s$ to $k$, and $k$ is invalid if such a ``diffusion path" does not exist. In order to implement a sale in practice, the seller should pick out all valid agents. We emphasize that though the set of participants evolves with the spread of the auction information, the scenario can be transformed into a static mechanism design setting: each agent directly submits her report to the seller; after receiving all reports, the seller then executes the mechanism among all valid agents. 

For convenience, let $F({\bf \theta'})$ denote all valid agents for a given $\theta'$. In addition, we use $T$ to denote the space of all possible transactions regarding to $N\cup\{s\}$, and $T(\theta')$ to denote the space of transactions construed by all valid agents $F({\bf \theta'})$. We now formally define the auction mechanisms with intermediaries.
\begin{defn}\label{int_auction}
	An \emph{auction mechanism} $M=(\pi,x)$ with intermediaries consists of an allocation policy $\pi: \Theta\rightarrow \mathcal{P}(T)$ and a payment policy $x=\{x_k:\Theta\rightarrow \mathcal{R}\}_{k\in N}$, where for all reported type profile $\theta'$, $\pi$ and $x$ satisfy the following constraints:
	\begin{itemize}
		\item [$1)$] $\pi(\theta')$ and $\{x_k(\theta')\}_{k\in N}$ are independent of the reports of invalid agents $N\setminus F(\theta')$;
		\item [$2)$] $\pi(\theta')\subseteq T(\theta')$ and $|\pi(\theta')|\le K$;
		\item [$3)$] $x_k(\theta')=0$ for all invalid agents.
	\end{itemize}
\end{defn}
Given all agents' reports $\theta'$, $\pi(\theta')$ outputs a set of transactions and $x_k(\theta')$ is the amount each agent $k$ pays.
To ensure the correctness of the model, we put three constraints on the auction policies. 
The first constraint requires that the mechanism should only sell the commodities among all valid agents. The second constraint indicates that the commodities cannot be oversold, and the last one states that the mechanism cannot charge an invalid agent.

Given a reported type profile $\theta'$ and an auction $(\pi,x)$, the transaction costs $C(\theta')$ for $\pi(\theta')$ is $\sum_{i\in I, k\in r_i}w_{ik} n_{ik}$, which can also be expressed as $\sum_{t\in \pi(\theta')}\sum_{(i,i+1)\in t}w_{ii+1}$. 
Here $(i,i+1)$ corresponds to two adjacent trading agents in transaction $t$. To conclude the transactions in $\pi(\theta')$, the seller should cover the transaction costs $C(\theta')$. The utility of agent $k$ in $(\pi,x)$ is defined as:
\begin{equation}
	u_k\big(\theta_k, {\bf \theta}^\prime, (\pi,x)\big)=
	\begin{cases}
		-x_k(\theta') &\text{$k\in I$,}\\
		z_k({\bf \theta}^\prime)\theta_k  - x_k({\bf \theta}^\prime) &\text{$k\in B$,}
	\end{cases}
\end{equation}
where $z_j(\theta')$ is an indicator with $z_j(\theta')=1$ for a winning buyer and $z_j(\theta')=0$ otherwise. For the seller, her utility, or revenue, is the sum of all buyers' payments minus the transaction costs, denoted by $R((\pi, x), {\bf \theta}^\prime) = \sum_{k\in N} x_k({\bf \theta}^\prime)-C(\theta')$.

Given an allocation $\pi(\theta')$, the social welfare is defined as the total utilities of all agents, which is identical to $\sum_{j\in B} z_j({\bf \theta}')v_j-C(\theta')$. We say an allocation policy is efficient if it maximizes the social welfare for all ${\bf \theta'}$.

\begin{defn}\label{efficient_allocation}
	An allocation policy $\pi^*$ is \emph{efficient} if for all ${\bf \theta'}$,
	\begin{equation}
		\pi^* \in {\arg\max}_{\pi^\prime \in \Pi} W^{\pi'}({\bf \theta'})
	\end{equation}
	where $\Pi$ is the set of all allocation polices defined in Definition \ref{int_auction} and $W^{\pi'}({\bf \theta'})$ is the social welfare achieved in $\pi'(\theta')$.
\end{defn}

By the revelation principle, we can restrict our attention on incentive compatible mechanisms w.l.o.g. 

\begin{defn}
	A mechanism $(\pi, x)$ is \emph{incentive compatible} (IC) if
	$u_k\big(\theta_k, (\theta_k,{\bf \theta}_{-k}^\prime), (\pi,x)\big) \geq u_k\big(\theta_k, (\theta_k^\prime, {\bf \theta}_{-k}^{\prime}), (\pi,x)\big)$ for all $k\in N$, all $\theta_k^\prime$, and all ${\bf \theta}_{-k}^\prime$.
\end{defn}
That is, submitting true type $\theta_k$ to the seller is a dominant strategy for all agent $k\in N$. To ensure that all agents are willing to stay in the auction, the mechanism should also be individually rational.

\begin{defn}
	A mechanism $(\pi,x)$ is \emph{individually rational} (IR) if $u_k\big(\theta_k, (\theta_k,{\bf \theta}_{-k}^\prime), (\pi,x)\big) \geq 0$ for all $k\in N$, and all ${\bf \theta}_{-k}^\prime$.
\end{defn}
This property, aka the participation constraint, guarantees that each agent will not receive a negative utility when revealing her type truthfully. In addition, a major reason for paying the intermediaries is to improve the auction outcomes, especially to increase the seller's revenue. Hence, an auction with an external subsidy is unsatisfying.


\begin{defn}
	A mechanism $(\pi, x)$ is \emph{weakly budget balanced} if $R({{M}}, {\bf \theta}') \geq 0$ for all ${\bf \theta}'$.
\end{defn}

The feature for auctions with intermediaries is that we must consider the connections between agents' reports. If an intermediary changes her neighbor set report, it may change the set of participants and further affect her own utility. This kind of ability is closely correlated with her position in the market, which makes the problem nontrivial. In the following contents, we focus on designing auction mechanisms that satisfy IR, IC and other desirable properties. To make the manuscript easy to follow, we present our results for the setting where all intermediaries are networked in a tree, and the results hold for arbitrary connections between intermediaries.

\section{Incentive Compatibility}\label{sec:ic}
Single-parameter domains \cite{archer2001truthful,auletta2004deterministic} refer to the settings where each agent's preference over allocations is characterized by a single parameter. Many classic settings, like single-item auctions or multi-unit auctions with unit-demand buyers, all fall into this domain. In our model, the intermediaries seek for profit maximization via information sharing and have no interests on the commodities. Moreover, the buyers are unit-demand and their preferences on the allocations can be characterized by their valuations. Therefore, our setting also falls into the category of single-parameter domains. It is well known that a normalized mechanism (i.e., losers pay zero) for single-parameter domains is IC and IR if and only if its allocation policy is value-monotonic, and the winner pays the minimum winning bid \cite{Myerson1981OptimalAD,archer2001truthful}. 


\begin{defn}\label{value-mono}
	An allocation policy $\pi$ is value-monotonic if for all ${\bf \theta'}$ and all valid buyer $j\in F(\theta')$ with $z_j(v_j', \theta'_{-j})=1$, we have $z_j(v_j'',\theta'_{-j})=1$ whenever $v_j''> v_j'$.
\end{defn}
In other words, an allocation policy is value-monotonic if a winning buyer cannot lose by posting a higher bid. 
Given a value-monotonic allocation policy, there exists a bid for each buyer, above/below which the buyer wins/loses following others' bids. This bid is called the buyer's critical bid which stands for the minimum bid for winning an item. 
\begin{defn}
	Given a value-monotonic $\pi$ and others' reports $ \theta'_{-j}$, the critical bid of a valid buyer $j$ is defined as $v^*_j(\theta')=\arg \min_{v_j'\in \mathcal{R}_{+}}\{z_j(v_j',\theta'_{-j})=1\}$.
\end{defn}

If the seller knows all potential buyers in advance, i.e., $r_s=B$, there is no need for an intermediary and the scenario reduces to a typical multi-unit auction setting. General theories for this setting have been well developed, see \cite{krishna2009auction} for an overview. For general cases, the diffusion strategies of intermediaries need be considered. To facilitate the spread of the auction information, the intermediaries should be rewarded for their diffusion. 

Recall that there exist two kinds of incentives in our model. That is, the buyers can misrepresent their valuations on the commodities and the intermediaries can selectively diffuse the sale information to her neighbors. Essentially, the latter incentive is quite different with the former. When characterizing incentive compatibility, technically they can be treated separately .


\begin{theorem}\label{ic_ir}
	An auction mechanism $(\pi,x)$ is incentive-compatible and individually rational if and only if 
	\begin{itemize}
		\item [1.] $\pi$ is value-monotonic and each winner pays her critical bid and the losing buyers pay zero;
		\item [2.] $x_k(r_k',\theta'_{-k})\le x_k(r_k'',\theta'_{-k})\le 0$ for all $k\in I$, all $\theta'=(r_k', \theta'_{-k})$ and all $r_i''\subseteq r_i'$.
	\end{itemize}
\end{theorem}
\footnotetext{We have moved all the proofs to the appendices, which can be found in the supplementary material.}
Theorem \ref{ic_ir} suggests that the seller can design the incentives for the intermediaries and the buyers independently to obtain truthfulness. This works if truthful reporting is the only desiderata for the seller. However, as we mentioned earlier the revenue is also a crucial consideration. If the mechanism does not improve the revenue, the seller is disincentivized to use it. Hence, we cannot treat the intermediaries and the buyers as unrelated individuals when designing auctions with performance advantages. 

\subsection{Non-Degenerated Mechanisms}
Given any IC and IR mechanism $M=(\pi, x)$, it is easy to construct another mechanism $M'=(\pi,x')$ where $x_j'=x_j$ for all buyers $j$ and $x_i'$ is a constant function for all intermediaries $i$. That is, $M'$ has the same allocation policy with $M$ but awards a constant money to each intermediary, e.g., setting $x_i'(\theta')=1$ for all intermediaries $i$ and all $\theta'$. We can check that $M'$ also satisfies the conditions listed in Theorem \ref{ic_ir}, i.e., $M'$ is IC and IR. However, $M'$ is not applicable in practice because it implicitly assumes that when an intermediary comes to a situation where she is indifference with sharing the auction information or not, she always chooses the former. This kind of mechanism gives no incentive to the intermediaries to do the sharing. For this reason, we also require the proposed mechanisms to satisfy the non-degenerated property, which is defined below.

\begin{defn}\label{non_trivial}
	$M$ is non-degenerated if for all intermediaries $i\in I$, $u_i(r_i,\theta, M)>u_i(r_i',\theta', M)\ge 0$ for some type profile $\theta=(r_i,\theta_{-i})$ and $\theta'=(r_i',\theta_{-i})$ with $r_i'\subset r_i$.
\end{defn}
For any non-degenerated auction mechanism, we actually create a situation in which an intermediary is uncertain about whether or not she gets rewarded before the auction, but she is sure that there is some type profile where sharing the auction information brings her a positive utility improvement. This property drives the intermediaries to propagate the auction information to others when facing uncertainties. In this work, we only consider non-degenerated auction mechanisms.

\section{Case Study: The VCG Mechanism}\label{sec:vcg}
As an introduction example, we apply the classic Vickrey-Clark-Groves (VCG) mechanism \cite{vickrey1961counterspeculation,clarke1971multipart,groves1973incentives} to our setting. In the VCG mechanism, the commodities are allocated in a way that maximizes the social welfare, i.e., it adopts the efficient allocation policy, and each agent is charged the social welfare decrease of others due to her participation. For convenience, we use $W^*(\theta'_{-k})$ to denote the social welfare without agent $k$'s participation and $W_{-k}^*(\theta')=W^*(\theta')-z_k(\theta')v_k'$ to denote the social welfare without considering agent $k$'s contribution ($v_k'$ is assumed as zero for all intermediaries). The formal definition of the VCG mechanism is given in the following.


\begin{framed}
	\noindent\textbf{The Vickrey-Clark-Groves Mechanism}\\
	\rule{\textwidth}{0.5pt}
	\begin{itemize}
		\item  \textbf{Allocation policy}: Given a reported type profile $\theta'$, apply the efficient allocation $\pi^*(\theta')$ (with random tie-breaking).
		\item  \textbf{Payment policy}: For each agent $k\in N$, 
		\begin{equation}
		x_k(\theta')=W^*(\theta'_{-k})-W_{-k}^*(\theta').
		\end{equation}
	\end{itemize}
\end{framed}


Recall that an agent could join in the sale only if she has received the sale information from someone who is already in the sale. As a result, each agent's participation is entangled with other agents' participation and diffusion. In our scenario, the expression $W^*(\theta'_{-k})$ in the payment policy actually represents for the maximum social welfare generated by all valid agents $F(\theta_{-k}')$.

\begin{prop}\label{vcg_outcome}
	The VCG mechanism is incentive compatible, individually rational, efficient and non-degenerated.
\end{prop}
With the assumption of separable costs, the efficient allocation $\pi^*(\theta')$ can be identified efficiently as follows. Let $T_j(\theta')$ denote the transaction between $s$ and a valid buyer $j$, and accordingly let $W_j(\theta')=v_j'-\sum_{(i,i+1)\in T_j(\theta')w_{ii+1}}$ be the social welfare created by $T_j(\theta')$. Given any reported type profile $\theta'$, we can compute $T_j(\theta')$ and $W_j(\theta')$ for all valid buyers, and $\pi^*(\theta')$ precisely corresponds to the set of transactions with top $K$ non-negative $W_j(\theta')$. 

To give an intuitive description, we apply the VCG mechanism in Figure \ref{example}. Suppose the seller is endowed with three commodities, and all agents act truthfully. The transactions with top three non-negative highest social welfare are $\{s,B,b_1\}$, $\{s,B,E,e_2\}$ and $\{s,A,D,d_1\}$. Accordingly, the social welfare achieved in the efficient allocation equals the summation of the social welfare in each transaction, which is $W^*(\theta)=(13-1)+(12-1-0)+(30-2-0)=51$, and the total transaction costs are $4$. Now consider the VCG payment of intermediary $B$. When $B$ does not participate in the auction (note that all agents rooted by $B$ cannot join in the sale as well), the efficient allocation includes transactions $\{s,A,D,d_1\}$, $\{s,A,D,d_2\}$ and $\{s,A,a_1\}$, and the social welfare is $W^*(\theta_{-B})=28+10+4=42$. According to the payment policy, $B$'s payment is $W^*(\theta_{-B})-W_{-B}^*(\theta)=42-(51-0)=-9$. In other words, the seller should reward $9$ to $B$ in the VCG mechanism. Similarly, we can compute the rewards for other intermediaries, which is $21$ for each $A$ and $D$, $0$ for $C$, and $1$ for $E$.
The winning buyers are $b_1$, $d_1$ and $e_2$ who pay $11, 12, 11$ to the seller respectively, and other buyers pay zero. Eventually, the seller's revenue is $11+12+11-9-2\cdot 21-0-1-4=-22$, where the last term $4$ represents for the transaction costs.

Although the VCG mechanism achieves the maximum social welfare, it is not weakly budget balanced, as shown in the above example. The underlying reason is the VCG mechanism treats the intermediaries the same as the buyers. As a result, it not only rewards the intermediaries to share the auction information, but also rewards them to bid truthfully, which is not necessary as they do not bid. To tackle the revenue issue, a natural approach is reducing the amounts paid to the intermediaries. Following this line, \cite{Li2018CustomerSI} proposed the customer sharing mechanism (CSM) for the case of single item. The mechanism firstly identifies the threshold neighborhoods for each intermediary, then uses them to define the payment. The threshold neighborhood is well-defined for the case of single item, but is not able to generalize to multiple items. In addition, there are attempts \cite{Zhao2018Multi,kawasaki2020strategy} studying multi-unit diffusion auctions with competitive revenue. However, the developed techniques are tailored for unweighted social network settings and are not applicable to our setting.
In the next section, we develop new techniques and propose a novel non-degenerated auction for selling multiple items regarding to our scenario, which is weakly budget balanced and outperforms the VCG mechanism on the seller's revenue.


\section{Critical Neighborhood Auction}\label{sec:cnm}
In this section, we propose the critical neighborhood auction (short for CNA) to solve the revenue issue of the VCG mechanism. The CNA not only incentivizes the intermediaries to share the auction information to all their neighbors, it also brings more revenue to the seller. In particular, we can prove that the revenue achieved in CNA is always no less than the revenue given by the VCG mechanism with/without using intermediaries. We first introduce the concepts of allocation graph and critical neighborhood, then give the formal definition of the mechanism. After that, we analyze the properties of CNA, and prove that it is IC, IR, non-degenerated, efficient and weakly budget balanced.

\subsection{Allocation Graph and Critical Neighborhood}
Allocation graph is a succinct representation of a given allocation. It is built by all transactions in the allocation, illustrating how the commodities flow from the seller to the intermediaries, then to the winners.

\begin{defn}[Allocation Graph]\label{all_graph}
	Given an allocation $\pi(\theta')$, the allocation graph $AG(\theta')$ is defined as the union of the transactions in $\pi(\theta')$, where an edge $(i,k)\in AG(\theta')$ if $(i,k)$ is included in $\pi(\theta')$.
\end{defn}

Once an allocation is identified, the construction of the corresponding allocation graph is straightforward. In particular, we use $AG^*(\theta')$ to denote the efficient allocation graph which is constructed by all transactions in the efficient allocation $\pi^*(\theta')$. For example, the subtree with bold arrows in Figure \ref{example} represents for the efficient allocation graph for three commodities. It should be noted that if an intermediary had shared the auction to all her neighbors, and she is not a node of $AG^*(\theta')$, then she is still an outsider for arbitrary sharing strategy.
Based on the concept of allocation graph, we now define the critical neighborhood for each intermediary, which is a key for the proposed mechanism.
\begin{defn}[Critical Neighborhood]\label{cn}
	Given a reported type profile $\theta'$ and an allocation $\pi(\theta')$, we call $\tilde{r}_i(\theta')=r_i'\cap AG(\theta')\cup I'_i$ the critical neighborhood of an intermediary $i$, where $I'_i$ are intermediaries in $r_i'$.
\end{defn}

Technically, $\tilde{r}_i(\theta')$ characterizes a set of invitations that are important for the allocation. To explain this, let $y$ denote the set of winning buyers in $\pi(\theta')$. If $i$ does not invite $\tilde{r}_i(\theta')$ to the sale, then all transactions to $y$, if there is any, {\it surely} do not pass $i$. In this sense, $\tilde{r}_i(\theta')$ are ``critical" for $i$ to reach the winners. Given any reported type profile $\theta'$, it is straightforward that for each intermediary $i$, her critical neighborhood $\tilde{r}_i(\theta')$ is existing and unique. In the example of Figure \ref{example}, agent $B$'s critical neighborhood is $\{b_1, C, E\}$. 

Based on the definition of critical neighborhood, we are ready to introduce our mechanism.
\subsection{The Mechanism}
The proposed mechanism applies an efficient allocation policy as the VCG does, and charges each intermediary the social welfare decrease of a particular transaction.  Comparing to the payments in the VCG mechanism, which can elicit both true bids from the buyers and true neighbor information from the intermediaries, the proposed mechanism only awards intermediaries for their diffusion efforts. The full description of the mechanism is given below, where $W^{(K)}(\theta')$ represents for the $K$th highest value in $\{W_j(\theta')\}_{j\in F(\theta')}$ and is defined as zero if the $K$th highest value is negative.

\begin{framed}
	\noindent\textbf{Critical Neighborhood Auction (CNA)}\\
	\rule{\textwidth}{0.5pt}
	\begin{itemize}
		\item \textbf{Allocation policy:} Given a reported type profile $\theta'$, apply the efficient allocation $\pi^*(\theta')$ (with random tie-breaking).
		\item \textbf{Payment policy:} The payment $x_k(\theta')$ is defined for each category of agents as follows:
		\begin{small}
			\begin{equation*}
				\begin{cases}
				W^{(K)}(\theta'_{-k})  - W^{(K)}(r_k'\setminus \tilde{r}_k(\theta'), \theta'_{-k}) &\text{$k\in I$,}\\
				W^*(\theta'_{-k})-W_{-k}^*(\theta') &\text{$k\in B$.}\\
				\end{cases}
			\end{equation*}
		\end{small}
	\end{itemize}
\end{framed}

In the critical neighborhood auction, each buyer pays her VCG payment. For each intermediary $i$, her payment is identical to the difference between the $K$th highest social welfare without her participation and the $K$th highest social welfare without diffusing the sale information to $\tilde{r}_i(\theta')$. The payment policy for each intermediary $i$ ensures that $1)$ $i$ could obtain a reward only if her diffusion improves the social welfare; and $2)$ more transactions $i$ is involved will lead to more profit for herself. As a result, the intermediaries are incentivized to diffuse the sale to all their neighbors to maximize their utilities. 

In the running example of the VCG mechanism, the seller needs pay $9$ to intermediary $B$, while she pays less in the CNA. If $B$ does not join in the sale, then the $3$rd highest social welfare is $4$ which is obtained in transaction $\{s,A,a_1\}$. If $B$ does not diffuse the sale to her critical neighborhood $\tilde{r}_B(\theta)=\{b_1,C, E\}$, the $3$rd highest social welfare becomes $7$ which is created by transaction $\{s,B,b_2\}$. Thus, $B$'s payment is $4-7=-3$, i.e., the seller only pays $B$ 3. Similarly, the seller pays $D$ $3$ and pays $0$ to other intermediaries. The buyers pay the same as they are in the VCG mechanism. Eventually, the seller's revenue by applying the CNA is $24$, which is a substantial improvement comparing to $-22$- the revenue achieved in the VCG mechanism.



\subsection{Theoretical Analysis}
Both the VCG and the CNA use the efficient allocation policy. For the efficient allocation policy, two observations should be highlighted. Firstly, the leaf nodes in the efficient allocation graph correspond to the winning buyers and each simple path from the root (i.e., the seller) to a leaf node corresponds to a winning transaction. Secondly, only the intermediaries in the efficient allocation graph can affect the allocation via strategic diffusion. For other intermediaries, their diffusion do not influence the transaction paths to the winning buyers, and thereby cannot affect the ultimate allocation, i.e., their payments and utilities are all zero.
Based on above observations and the concept of critical neighborhood, we can prove that the CNA is IC and IR.

\begin{theorem}\label{ir_ic}
	The CNA is efficient, individually rational, incentive-compatible and non-degenerated.
\end{theorem}

It should be noted that the critical neighborhood is a key component for proving Theorem \ref{ir_ic}. Our following results show that it is also crucial for improving the seller's revenue. In the following, we analyze the revenue performance of the CNA. By comparing the payments of the VCG and the CNA, we have the following result.

\begin{prop}
	$R(\text{CNA},\theta')\ge R(\text{VCG},\theta')$ for all $\theta'$.
\end{prop}

Recall that each agent' payment is correlated with the underlying market structures, therefore it is hard to evaluate the seller's revenue by directly summing up all agents' payments. Instead, we first decompose the payments into several groups; then provide a lower bound for each group's payments; by aggregating the lower bound in each group, we can show that the CNA also outperforms the VCG mechanism without intermediaries (i.e., applying the VCG mechanism only in the neighboring buyers of the seller $s$). Figure \ref{str_proof} demonstrates the computation of the revenue w.r.t. Figure \ref{example}. 

Before illustrating our result, we first introduce several useful notations. Given an efficient allocation graph $AG^*(\theta')$ and a leaf node $j$, let $T^*_j$ be the $s$-$j$ path in $AG^*(\theta')$ (for convenience we remove the seller $s$ from $T^*_j$). In addition, let $X(T^*_j)=\sum_{k\in T^*_j}x_k(\theta')$ denote the path payment and $C(T^*_j)=\sum_{(i,i+1)\in T^*_j}w_{ii+1}$ denote the associated path cost. We can prove the following for all $T^*_j$ in $AG^*(\theta')$.

\begin{lemma}
	For all $T^*_j$ in $AG^*(\theta')$, $X(T^*_j)-C(T^*_j)\ge W^{(K)}(\theta'_{-1})$, where $1$ is the first agent in $T^*_j$.
\end{lemma}

Denote by $n_k$ the number of transactions including agent $k$ in the allocation $\pi^*(\theta')$. Clearly, $n_k=1$ for all winning buyers and $n_k\ge 1$ for all intermediaries in $\pi^*(\theta')$. Since $x_k(\theta')\le 0$ for all intermediaries, we have \begin{small}
\begin{equation}
	\sum_{k\in T^*_j}\frac{x_k(\theta')}{n_k}\ge \sum_{k\in T^*_j}x_k(\theta')=X(T^*_j).
\end{equation} \end{small}
Given the fact that \begin{small}
\begin{equation}
	\sum_{j\in \text{leaf nodes}}(\sum_{k\in T^*_j}\frac{x_k(\theta')}{n_k}-C(T^*_j))=\sum_{k\in AG^*(\theta)}x_k(\theta')-C(\theta'),
\end{equation}\end{small}
we can directly prove the following.

\begin{theorem}\label{last}
	$R(\text{CNA},\theta')\ge R(\text{VCG-WI}, \theta')\ge 0$ for all $\theta'$, where $\text{VCG-WI}$ refers to the VCG Without Intermediaries.
\end{theorem}

\begin{figure}[t]
	\centering
	\includegraphics[width=3.4in]{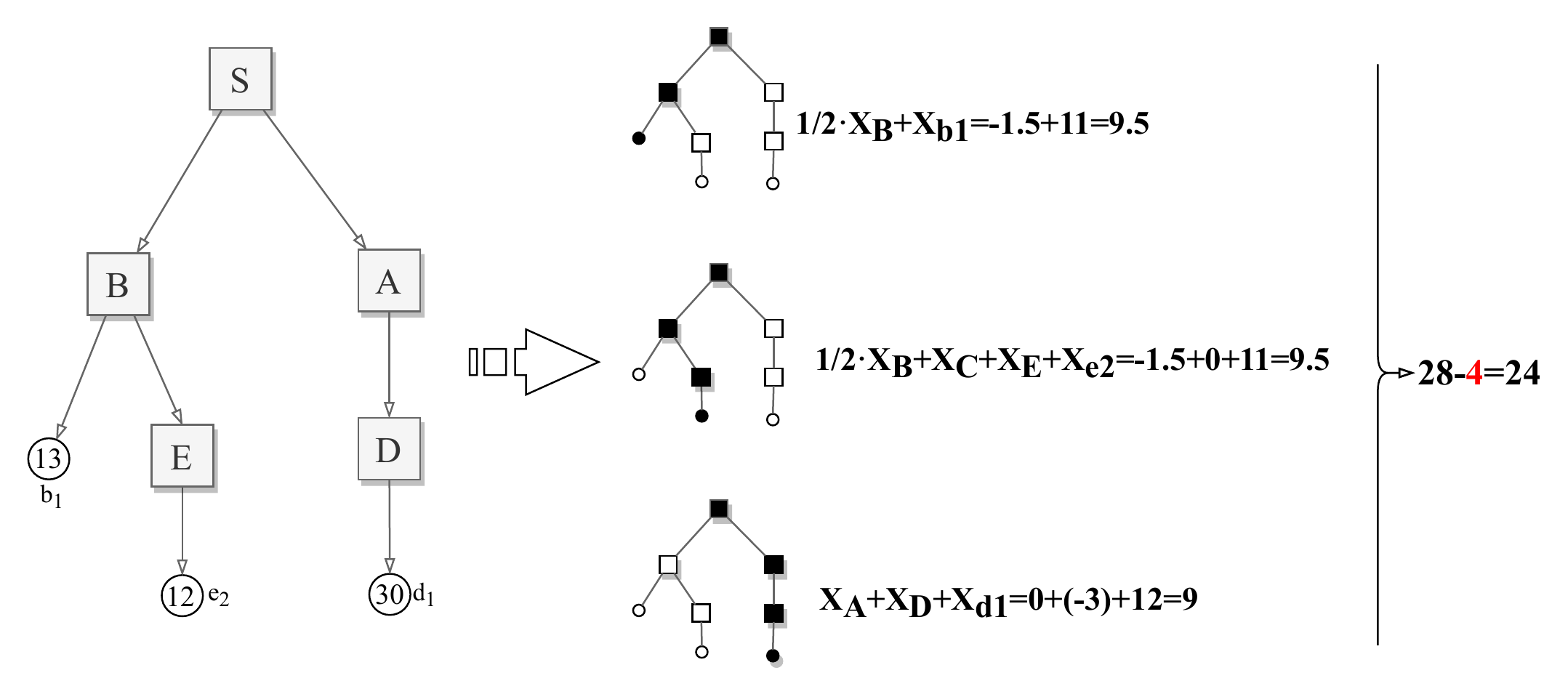}
	\caption{The decomposition and aggregation process for computing the seller's revenue, where the left corresponds to the efficient allocation graph for three commodities and the red number in the right denotes the transaction costs.}\label{str_proof}
\end{figure}
For example, if we apply $\text{VCG-WI}$ in Figure \ref{example}, the three commodities will be allocated to buyers $s_2$, $s_3$, $s_4$, and each winner pays $1$ to the seller, i.e., the seller's revenue is $3<24$. Theorem \ref{last} demonstrates the advantages of diffusion auctions over traditional auctions. 
Though the seller has to pay intermediaries to reach more buyers with high valuations, the benefits from these additional buyers can cover the loss. 

\section{Conclusions}\label{sec:con}
This work considers an auction market with intermediaries, and proposes a diffusion auction, called critical neighborhood auction, for a seller selling multiple homogeneous commodities to unit-demand buyers.
The mechanism not only achieves the maximum social welfare within the market, but also increases the seller's revenue comparing to traditional auctions.
There are many problems worth further investigation. For example, this work assumes that the commodities are homogeneous and the buyers are unit-demand, which simplifies the model and the analysis. An extension is to consider heterogeneous commodities where the buyers are single-minded or additive.
In addition, a marketplace often gathers substantial amounts of sellers and buyers with various supplies and demands, designing efficient markets to incentivize information spreading, and to arrange the transactions globally is also an important future work.

\newpage
\bibliographystyle{named}
\bibliography{ijcai22}

\begin{thebibliography}{}

\bibitem[\protect\citeauthoryear{Archer and Tardos}{2001}]{archer2001truthful}
Aaron Archer and {\'E}va Tardos.
\newblock Truthful mechanisms for one-parameter agents.
\newblock In {\em FOCS}, pages 482--491, 2001.

\bibitem[\protect\citeauthoryear{Auletta \bgroup \em et al.\egroup
  }{2004}]{auletta2004deterministic}
Vincenzo Auletta, Roberto De~Prisco, Paolo Penna, and Giuseppe Persiano.
\newblock Deterministic truthful approximation mechanisms for scheduling
  related machines.
\newblock In {\em STACS}, pages 608--619, 2004.

\bibitem[\protect\citeauthoryear{Clarke}{1971}]{clarke1971multipart}
Edward~H Clarke.
\newblock Multipart pricing of public goods.
\newblock {\em Public Choice}, 11(1):17--33, 1971.

\bibitem[\protect\citeauthoryear{Emek \bgroup \em et al.\egroup
  }{2011}]{emek2011mechanisms}
Yuval Emek, Ron Karidi, Moshe Tennenholtz, and Aviv Zohar.
\newblock Mechanisms for multi-level marketing.
\newblock In {\em Proc. of the 12th ACM Conf. on Electronic Commerce}, pages
  209--218. ACM, 2011.

\bibitem[\protect\citeauthoryear{Groves}{1973}]{groves1973incentives}
Theodore Groves.
\newblock Incentives in teams.
\newblock {\em Econometrica}, 41(4):617--631, 1973.

\bibitem[\protect\citeauthoryear{Guo and Hao}{2021}]{guo2021emerg}
Yuhang Guo and Dong Hao.
\newblock Emerging methods of auction design in social networks.
\newblock In {\em Proc. of the 30th Int. Joint Conf. on Artificial
  Intelligence}, pages 4434--4441, 2021.

\bibitem[\protect\citeauthoryear{Jeong and Lee}{2020}]{jeong2020groupwise}
Seungwon~(Eugene) Jeong and Joosung Lee.
\newblock The groupwise-pivotal referral mechanism: Core-selecting referral
  strategy-proof mechanism.
\newblock In {\em Available at SSRN: https://ssrn.com/abstract=3574093}, 2020.

\bibitem[\protect\citeauthoryear{Kawasaki \bgroup \em et al.\egroup
  }{2020}]{kawasaki2020strategy}
Takehiro Kawasaki, Nathana{\"e}l Barrot, Seiji Takanashi, Taiki Todo, and
  Makoto Yokoo.
\newblock Strategy-proof and non-wasteful multi-unit auction via social
  network.
\newblock In {\em Proc. of the AAAI Conf. on Artificial Intelligence}, pages
  2062--2069, 2020.

\bibitem[\protect\citeauthoryear{Kawasaki \bgroup \em et al.\egroup
  }{2021}]{kawasaki2021mechanism}
Takehiro Kawasaki, Ryoji Wada, Taiki Todo, and Makoto Yokoo.
\newblock Mechanism design for housing markets over social networks.
\newblock In {\em Proc. of the 20th Int. Conf. on Autonomous Agents and
  MultiAgent Systems}, pages 692--700, 2021.

\bibitem[\protect\citeauthoryear{{Kleinberg} and
  {Raghavan}}{2005}]{kleinberg2005query}
J.~{Kleinberg} and Prabhakar {Raghavan}.
\newblock Query incentive networks.
\newblock In {\em Proc. of the 46th Annual IEEE Symposium on Foundations of
  Computer Science}, pages 132--141, 2005.

\bibitem[\protect\citeauthoryear{Krishna}{2009}]{krishna2009auction}
Vijay Krishna.
\newblock {\em Auction theory}.
\newblock Academic Press, 2009.

\bibitem[\protect\citeauthoryear{Leskovec \bgroup \em et al.\egroup
  }{2006}]{leskovec2007dynamics}
Jure Leskovec, Lada~A. Adamic, and Bernardo~A. Huberman.
\newblock The dynamics of viral marketing.
\newblock In {\em Proc. of the 7th ACM Conf. on Electronic Commerce}, pages
  228--237, 2006.

\bibitem[\protect\citeauthoryear{Li \bgroup \em et al.\egroup
  }{2017}]{li2017mechanism}
Bin Li, Dong Hao, Dengji Zhao, and Tao Zhou.
\newblock Mechanism design in social networks.
\newblock In {\em Proc. of the 31st AAAI Conference on Artificial
  Intelligence}, pages 586--592, 2017.

\bibitem[\protect\citeauthoryear{Li \bgroup \em et al.\egroup
  }{2018}]{Li2018CustomerSI}
Bin Li, Dong Hao, Dengji Zhao, and Tao Zhou.
\newblock Customer sharing in economic networks with costs.
\newblock In {\em Proc. of the 27th Int. Joint Conf. on Artificial
  Intelligence}, pages 368--374, 2018.

\bibitem[\protect\citeauthoryear{{Li} \bgroup \em et al.\egroup
  }{2019}]{li2019graph}
Bin {Li}, Dong {Hao}, Dengji {Zhao}, and Makoto {Yokoo}.
\newblock Diffusion and auction on graphs.
\newblock In {\em Proc. of the 28th Int. Joint Conf. on Artificial
  Intelligence}, pages 435--441, 2019.

\bibitem[\protect\citeauthoryear{Li \bgroup \em et al.\egroup
  }{2022}]{LI2022103631}
Bin Li, Dong Hao, Hui Gao, and Dengji Zhao.
\newblock Diffusion auction design.
\newblock {\em Artificial Intelligence}, 303:103631, 2022.

\bibitem[\protect\citeauthoryear{Liu \bgroup \em et al.\egroup
  }{2021}]{liu2021budget}
Xiang Liu, Weiwei Wu, Minming Li, and Wanyuan Wang.
\newblock Budget feasible mechanisms over graphs.
\newblock In {\em Proc. of the AAAI Conf. on Artificial Intelligence}, pages
  5549--5556, 2021.

\bibitem[\protect\citeauthoryear{Moustafa \bgroup \em et al.\egroup
  }{2021}]{moustafa2021diffusion}
Ahmed Moustafa, Pankaj Mishra, and Nagoya~Kogyo Daigaku.
\newblock A diffusion mechanism for multi-unit commodity allocation in economic
  networks.
\newblock {\em Electronic Commerce Research and Applications}, page 101078,
  2021.

\bibitem[\protect\citeauthoryear{Myerson}{1981}]{Myerson1981OptimalAD}
Roger~B. Myerson.
\newblock Optimal auction design.
\newblock {\em Mathematics of Operations Research}, 6:58--73, 1981.

\bibitem[\protect\citeauthoryear{Vickrey}{1961}]{vickrey1961counterspeculation}
William Vickrey.
\newblock Counterspeculation, auctions, and competitive sealed tenders.
\newblock {\em The Journal of Finance}, 16(1):8--37, 1961.

\bibitem[\protect\citeauthoryear{Yang \bgroup \em et al.\egroup
  }{2018}]{yang2018incentive}
Liu Yang, Hongbin Zhu, Haifeng Wang, Hua Qian, and Yang Yang.
\newblock Incentive propagation mechanism of computation offloading in
  fog-enabled d2d networks.
\newblock In {\em the IEEE 23rd Int. Conf. on Digital Signal Processing}, pages
  1--4, 2018.

\bibitem[\protect\citeauthoryear{{Zhang} \bgroup \em et al.\egroup
  }{2020a}]{zhang2020collaborative}
Wen {Zhang}, Yao {Zhang}, and Dengji {Zhao}.
\newblock Collaborative data acquisition.
\newblock In {\em Proc. of the 19th Int. Conf. on Autonomous Agents and
  MultiAgent Systems}, pages 1629--1637, 2020.

\bibitem[\protect\citeauthoryear{{Zhang} \bgroup \em et al.\egroup
  }{2020b}]{zhang2019incentivize}
Wen {Zhang}, Dengji {Zhao}, and Yao {Zhang}.
\newblock Incentivize diffusion with fair rewards.
\newblock In {\em Proc. of the 24th European Conf. on Artificial Intelligence},
  pages 251--258, 2020.

\bibitem[\protect\citeauthoryear{{Zhang} \bgroup \em et al.\egroup
  }{2020c}]{zhang2020sybil}
Yao {Zhang}, Xiuzhen {Zhang}, and Dengji {Zhao}.
\newblock Sybil-proof answer querying mechanism.
\newblock In {\em Proc. of the 29th Int. Joint Conf. on Artificial
  Intelligence}, pages 422--428, 2020.

\bibitem[\protect\citeauthoryear{Zhao \bgroup \em et al.\egroup
  }{2018}]{Zhao2018Multi}
Dengji Zhao, Bin Li, Junping Xu, Dong Hao, and Nicholas~R. Jennings.
\newblock Selling multiple items via social networks.
\newblock In {\em Proc. of the 17th Int. Conf. on Autonomous Agents and
  MultiAgent Systems}, pages 68--76, 2018.

\bibitem[\protect\citeauthoryear{Zhao}{2021}]{zhao2021social}
Dengji Zhao.
\newblock Mechanism design powered by social interactions.
\newblock In {\em Proc. of the 20th Int. Conf. on Autonomous Agents and
  MultiAgent Systems}, pages 63–--67, 2021.

\end{thebibliography}
	\appendix
\section*{Appendix}
\setcounter{theorem}{0}
\setcounter{prop}{0}
\setcounter{lemma}{0}
\section{Omitted Proofs for Section 3}

\begin{theorem}\label{app_ic_ir}
	An auction mechanism $(\pi,x)$ is incentive-compatible and individually rational if and only if 
	\begin{itemize}
		\item [1.] $\pi$ is value-monotonic and each winner pays her critical bid and the losing buyers pay zero;
		\item [2.] $x_k(r_k',\theta'_{-k})\le x_k(r_k'',\theta'_{-k})\le 0$ for all $k\in I$, all $\theta'=(r_k', \theta'_{-k})$ and all $r_i''\subseteq r_i'$.
	\end{itemize}
\end{theorem}
\begin{proof}
	The part for the buyers is the traditional IC property of single-parameter domains, which has been observed in many times, see [Myerson, 1981] for an example; and the part for the intermediaries can be directly obtained by the definition of IC and IR. 
\end{proof}
\section{Omitted Proofs for Section 4}
\begin{prop}
	The VCG mechanism is incentive-compatible, individually rational, efficient and non-degenerated.
\end{prop}
\begin{proof}
	The VCG mechanism adopts the efficient allocation policy, which is value-monotonic. We next prove that it is IC and IR. According to the payment policy, each intermediary $k$ pays $W^*(\theta'_{-k})-W_{-k}^*(\theta')$, which is identical to $W^*(\theta'_{-k})-W^*(\theta')$ as the intermediaries do not bid. The former term $W^*(\theta'_{-k})$ is independent of $r_k'$. The latter term $W^*(\theta')$ represents the maximum social welfare achieved in $F(\theta')$. Sharing the sale information to more neighbors will get more agents to participate in the auction, which will potentially increase $W^*(\theta')$. Therefore, $W^*(\theta')$ is non-decreasing with $r_k'$ and the VCG mechanism is IC and IR for all intermediaries according to Theorem \ref{app_ic_ir}. For each losing buyer $k$, she is not in the efficient allocation $\pi^*(\theta')$, and thus  $W^*(\theta'_{-k})=W_{-k}^*(\theta')=W^*(\theta')$, resulting in a zero payment. For each winner $k$, her payment is $W^*(\theta'_{-k})-W_{-k}^*(\theta')$. By subtracting the common part $W^*(\theta')-W_k^*(\theta')$ in both $W^*(\theta'_{-k})$ and $W_{-k}^*(\theta')$, we can simplify $k$'s payment as $W^{(K)}(\theta'_{-k})+\sum_{(i,i+1)\in t_{k}^*}w_{ii+1}$, where $T_{k}^*$ is the transaction associated with $k$ in $\pi^*(\theta')$, and $W^{(K)}(\theta'_{-k})$ denotes the social welfare obtained in the $K$th highest transaction under $\theta'_{-k}$, which is defined as $0$ if the $K$th highest social welfare is negative. Note that the allocation policy is efficient, and thus the value $W^{(K)}(\theta'_{-k})+\sum_{(i,i+1)\in T_{k}^*}w_{ii+1}$ is exactly the critical bid of $k$ for winning one item under $\theta'$. Hence, the VCG mechanism is IC and IR for all buyers. In addition, for all intermediaries $k$, her utility is strictly improved as long as her diffusion increases the social welfare $W^*(\theta')$, which is possible as long as she is included in the efficient allocation, i.e., the VCG mechanism is non-degenerated. 
\end{proof}
\section{Omitted Proofs for Section 5}
\begin{theorem}
	The CNA is efficient, individually rational, incentive-compatible and non-degenerated.
\end{theorem}
\begin{proof}
	The CNA is efficient according to the definition. We next show it is IC and IR, which is sufficient to check the IC and IR properties for all intermediaries as the buyers pay the same as they are in the VCG mechanism. For each intermediary $k$, her payment is $W^{(K)}(\theta'_{-k})  - W^{(K)}(r_k'\setminus \tilde{r}_k(\theta'), \theta'_{-k})$ and her utility is $W^{(K)}(r_k'\setminus \tilde{r}_k(\theta'), \theta'_{-k})-W^{(K)}(\theta'_{-k})$. According to the definition of critical neighborhood, we have that $F(\theta'_{-k})\subset F(r_k'\setminus \tilde{r}_k(\theta'), \theta'_{-k})$, and therefore $k$'s utility is non-negative. Next, we prove that diffusing the sale information to all neighbors maximizes each intermediary's utility. Let $\tilde{r}_k$, $\tilde{r}_k'$ be the critical neighborhood of $k$ under strategies $r_k$ and $r_k'\subset r_k$, respectively. To prove IC, it is enough to show that $W^{(K)}(r_k\setminus \tilde{r}_k, \theta'_{-k})\le W^{(K)}(r_k'\setminus \tilde{r}_k', \theta'_{-k})$ for all $r_k'\subset r_k$, which is equivalent to show that $\tilde{r}_k\subseteq \{r_k\setminus r_k'\}\cup \tilde{r}_k'$ holds for all $r_k'\subset r_k$. 
	
	For every element $e\in \tilde{r}_k$, if $e$ locates in $AG^*(\theta')$, then according to the definition of critical neighborhood and the efficient allocation policy, we know that $e$ is either in $r_k\setminus r_k'$ or $\tilde{r}_k'$; otherwise, $e$ is an intermediary according to the definition of critical neighborhood, which also leads to the conclusion that $e$ is either in $r_k\setminus r_k'$ or $\tilde{r}_k'$. The above reasoning indicates that no matter to which neighbor set $k$ diffuses the sale information, we always have that $\tilde{r}_k\subseteq \{r_k\setminus r_k'\}\cup \tilde{r}_k'$. Therefore, in the CNA the intermediaries will diffuse the sale information to all their neighbors to maximize utilities, i.e., the CNA is IC for all intermediaries. Moreover, for each intermediary $k$, she will get a strict utility improvement whenever the new created winning transactions after removing her critical neighborhood pass $k$. That is, CNA is non-degenerated.
\end{proof}
\begin{prop}
	$R(\text{CNA},\theta')\ge R(\text{VCG},\theta')$ for all $\theta'$.
\end{prop}
\begin{proof}
	It is enough to show that $x_i^{cna}(\theta') \ge x_i^{vcg}(\theta')$ for all intermediary $i$, which is true as
	\setcounter{equation}{0}
	\begin{align}
		x_i^{cna}(\theta')&=W^{(K)}(\theta'_{-i})  - W^{(K)}(r'_i\setminus \tilde{r}_i(\theta'), \theta'_{-i})\\
		&\ge W^{(K)}(\theta'_{-i})  - W^{(K)}(\theta')\\
		&\ge W^*(\theta'_{-i})-W^*(\theta')\\
		&=W^*(\theta'_{-i})-W_{-i}^*(\theta')=x_i^{vcg}(\theta'),
	\end{align} 
	where $(1)$ is from the fact that $F(r'_i\setminus \tilde{r}_i(\theta'), \theta'_{-i})\subseteq F(\theta')$, $(3)$ is due to the efficient allocation policy and $(4)$ is from the observation that $i$ can be viewed as a losing buyer.
\end{proof}
\begin{lemma}\label{path_payment}
	For all $T^*_j$ in $AG^*(\theta')$, $X(T^*_j)-C(T^*_j)\ge W^{(K)}(\theta'_{-1})$, where $1$ represents for the first agent in $T^*_j$.
\end{lemma}
\begin{proof}
	For technical convenience, we use $\{1,\cdots,j-1,j\}$ to denote $T^*_j$. For any $T^*_j$, we have that
	\begin{small}
		\begin{align*}
			X(T^*_j)&=\sum_{k\in T^*_j} x_k(\theta')= \sum_{k\in T^*_j\setminus\{j\}} x_k(\theta')+x_j(\theta')\\
			&=\sum_{k\in T^*_j\setminus\{j\}} (W^{(K)}(\theta'_{-k})  - W^{(K)}(r'_k\setminus \tilde{r}_k(\theta'), \theta'_{-k}))+\\&W^*(\theta'_{-j})-W_{-j}^*(\theta').
		\end{align*}
	\end{small}
	
	By subtracting $W^*(\theta')-W_j^*(\theta')$ in both $W^*(\theta'_{-j})$ and $W_{-j}^*(\theta')$, $W^*(\theta'_{-j})-W_{-j}^*(\theta')$ is simplified as $W^{(K)}(\theta'_{-j})+C(T^*_j)$. Therefore, $X(T^*_j)$ can be further denoted by
	\begin{small}
		\begin{equation}
			\begin{aligned}
				X(T^*_j)=&W^{(K)}(\theta'_{-1})+C(T^*_j)+\sum_{k\in T^*_j\setminus\{j\}} (W^{(K)}(\theta'_{-{k+1}})-\\&W^{(K)}(r'_k\setminus \tilde{r}_k(\theta'), \theta'_{-k})).
			\end{aligned}
	\end{equation}\end{small}
	
	Since the market has a tree structure, we know that $F(\theta'_{-k})\subset F(r'_k\setminus \tilde{r}_k(\theta'), \theta'_{-k})\subseteq F(\theta'_{-{k+1}})$, which means that for all $k\in T^*_j\setminus\{j\}$,  \begin{small}
		\begin{equation}
			W^{(K)}(\theta'_{-{k+1}})-W^{(K)}(r'_k\setminus \tilde{r}_k(\theta'), \theta'_{-k})\ge 0.
	\end{equation}\end{small}
	
	Combined with $(5)$ and $(6)$, we have that $X(T^*_j)\ge W^{(K)}(\theta'_{-1})+C(T^*_j)\Rightarrow X(T^*_j)-C(T^*_j)\ge W^{(K)}(\theta'_{-1})$.
\end{proof}
\begin{theorem}
	$R(\text{CNA},\theta')\ge R(\text{VCG-WI}, \theta')\ge 0$ for all $\theta'$, where $\text{VCG-WI}$ refers to the VCG Without Intermediaries.
\end{theorem}
\begin{proof}
	Given any reported type profile $\theta'$, the seller's revenue achieved in CNA is $R(\text{CNA},\theta')=\sum_{k\in N}x_k(\theta')-C(\theta')$. Since the intermediaries who are not in $AG^*(\theta')$ cannot change $AG^*(\theta')$ via strategic reports, their payments are all zero according to the payment policy of CNA. Hence, $R(\text{CNA},\theta')$ can be expressed as $\sum_{k\in AG^*(\theta')}x_k(\theta')-C(\theta')$. Based on Lemma \ref{path_payment}, we have that 
	\begin{small}
		\begin{align*}
			R(\text{CNA},\theta')&=\sum_{k\in AG^*(\theta')}x_k(\theta')-C(\theta')\\
			&=\sum_{j\in \text{leaf nodes}}(\sum_{k\in T^*_j}\frac{x_k(\theta')}{n_k}-C(T^*_j))\\
			&\ge \sum_{j\in \text{leaf nodes}}(\sum_{k\in T^*_j}x_k(\theta')-C(T^*_j))\\
			&=\sum_{j\in \text{leaf nodes}}(X(T^*_j)-C(T^*_j))\\
			&\ge \sum_{j\in \text{leaf nodes}}W^{(K)}(\theta'_{-1}).
		\end{align*}
	\end{small}
	
	For the multi-unit unit-demand setting, the traditional VCG mechanism will allocate the items to the buyers with top $K$ highest bids, and charges each of them the $K+1$th highest bid. Since the traditional VCG mechanism does not pay the intermediaries, the intermediaries have no incentive to invite others to the sale. That is, when applying the traditional VCG mechanism in our setting, only the buyers reached by the seller, namely buyers in $r_s$, are able to participate in the auction. Therefore, $R(\text{VCG-WI},\theta')$ is identical to $K\cdot v^{(K+1)}$, where $v^{(K+1)}$ denotes the $K+1$th highest bid in $r_s$. 
	
	Recall that $r_s\subseteq F(\theta'_{-1})\cup \{1\}$ for every $T^*_j$, we have that $W^{(K)}(\theta'_{-1})$, the $K$th highest social welfare without agent $1$'s participation, is no less than the $K+1$th highest social welfare in $F(\theta'_{-1})\cup \{1\}$, which is at least $v^{(K+1)}$. Therefore, $R(\text{CNA},\theta')\ge \sum_{j\in \text{leaf nodes}}W^{(K)}(\theta'_{-1})\ge \sum_{j\in \text{leaf nodes}}v^{(K+1)}=K\cdot v^{(K+1)}=R(\text{VCG-WI},\theta')$, which completes the proof.
\end{proof}
\end{document}